\documentclass[11pt,draftcls,onecolumn]{IEEEtran}
\usepackage{amsmath}
\usepackage{amssymb}
\usepackage{amsthm}
\usepackage{graphicx}
\usepackage{cases}
\usepackage{subfig}
\usepackage{enumerate}
\usepackage{algorithmic}
\usepackage{algorithm}
\usepackage{datetime}
\usepackage{booktabs}
\usepackage{multirow}
\usepackage{acronym}
\usepackage{ifthen}
\usepackage{multicol}

\newtheorem{theorem}{Theorem}[section]

\newtheorem{definition}[theorem]{Definition}

\newtheorem{assumption}[theorem]{Assumption}

\graphicspath{{./figures/}} 

\acrodef{KCF}{Kronecker Canonical Form}
\acrodef{MAPE}{Maximum a Posterior}
\acrodef{MLE}{Maximum Likelihood Estimate}

\author{Ali A. Al-Matouq$^{a}$ $^{\ast}$\thanks{$^\ast$Corresponding author. %
Email: aalmatou@mines.edu \vspace{6pt}}%
\\\vspace{6pt} $^{a}${\em{Department of Electrical Engineering and Computer Science,}}$^{b}${\em{Department of Applied Mathematics and Statistics, Colorado School of Mines 1600 Illinois St., Golden, CO 80401}}%
\\\vspace{3pt} }

\title{Derivation of the Maximum a Posteriori Estimate for Discrete Time Descriptor Systems}

\begin{document}

\maketitle 

\begin{abstract}
In this report a derivation of the MAP state estimator objective function for general (possibly non-square) discrete time causal/non-causal descriptor systems is presented.  The derivation made use of the Kronecker Canonical Transformation to extract the prior distribution on the descriptor state vector so that Maximum a Posteriori (MAP) point estimation can be used.  The analysis indicates that the MAP estimate for index 1 causal descriptor systems does not require any model transformations and can be found recursively.  Furthermore, if the descriptor system is of index 2 or higher and the noise free system is causal, then the MAP estimate can also be found recursively without model transformations provided that model causality is accounted for in designing the stochastic model.  
\end{abstract}

\begin{IEEEkeywords}
Descriptor Systems, Maximum a Posteriori Estimate, Maximum Likelihood Estimate
\end{IEEEkeywords}

\section{Maximum a Posteriori Estimation for Discrete Time Descriptor Systems} \label{ch3}

\subsection{Introduction}
The first objective of this chapter is to find the maximum a posterior (MAP) estimate for the state vector sequence $x_{k}$ given a stochastic discrete time descriptor system model (SDTDS), noisy measurements and an informative prior for $Ex_{0}$ as follows:
\begin{align}\label{dts1_ch3}
Ex_{k+1}=&Ax_{k}+Bu_{k}+Fw_{k} \\
y_{k}=&Hx_{k}+v_{k}\label{meas_ch3} \\
Ex_{0}&\sim  \mathtt{N}(\bar{r}_{0},P_{0}) \label{prior}
\end{align}
where $\mathtt{N}(\bar{r}_{0},P_{0})$ denotes a normally distributed random variable with mean $\bar{r}_{0}$ and variance $P_{0}$.  We assume here that only the sequence $u_{k}$ is deterministic and all other sequences are random.  The input disturbance sequence $w_{k}\in\mathbb{R}^{p}$ and the measurement noise sequence $v_{k}\in\mathbb{R}^{q}$ are i.i.d. normal random sequences; $w_{k}\sim \mathtt{N}(0,I_{q})$ and $v_{k}\sim \mathtt{N}(0,R)$, where $R\succ 0$.  Furthermore, the random variables $Ex_{0},w_{k},v_{k}$ are assumed uncorrelated between each other.  The matrices $E, A\in \mathbb{R}^{n_{eq}\times n}$, $B\in \mathbb{R}^{n_{eq}\times j}$, $y\in \mathbb{R}^{m}$, $H\in \mathbb{R}^{m\times n}$ and $F\in \mathbb{R}^{n_{eq}\times p}$.

The maximum a posteriori estimate of $x_{k}$ is defined as the mode of the posterior distribution denoted by $\hat{\mathbf{x}}^{map}$ and given by:
\begin{align}
\hat{\mathbf{x}}^{map}:=&\mbox{arg}\max_{\mathbf{x}} p_{\mathbf{x}|\mathbf{y}}(\mathbf{x}|\mathbf{y})\nonumber \\
= & \mbox{arg}\max_{\mathbf{x}} p_{\mathbf{y}|\mathbf{x}}(\mathbf{y}|\mathbf{x})\frac{p_{\mathbf{x}}(\mathbf{x})}{p_{\mathbf{y}}(\mathbf{y})}= \mbox{arg}\max_{\mathbf{x}}(\log p_{\mathbf{y}|\mathbf{x}}(\mathbf{y}|\mathbf{x}) + \log p_{\mathbf{x}}(\mathbf{x}))
\end{align}
where, $\hat{\mathbf{x}}^{map}=\{\hat{x}_{k}^{map}\}_{k=0}^{T}$, $\mathbf{x}=\{x_{k}\}_{k=0}^{T}$, $\mathbf{y}=\{y_{k}\}_{k=0}^{T}$ and $\hat{x}_{k}^{map}$ is the MAP estimate at time $k$.  Linear descriptor systems define $x_{k}$ implicitly and hence the prior distribution $p_{\mathbf{x}}(\mathbf{x})$ can not be found directly from the stochastic descriptor system given in \eqref{dts1_ch3}.  A proceeding step is needed to convert the stochastic descriptor system to a format that reveals the prior distribution on $x_{k}$.  On the other hand, the constrained maximum likelihood estimate $\hat{x}_{k}^{ml}$ is found by treating the state sequence $x_{k}$ as a parameter and the estimates are obtained by maximizing the likelihood function:
\begin{align}
&\hat{x}_{k}^{ml}:=\mbox{arg}\max_{\mathbf{x}} \mathtt{L}(\mathbf{x}|\mathbf{y})=\mbox{arg}\max_{\mathbf{x}} p_{\mathbf{y}|\mathbf{x}}(\mathbf{y}|\mathbf{x})\nonumber \\
&\mbox{subject to}~~\mathbf{x}\in \mathtt{C}
\end{align}
where the constraint $\mathbf{x}\in \mathtt{C}$ forms the prior information about the parameter $\mathbf{x}$.  In \cite{Nikoukhah1992}, the input sequence $Bu_{k}$ and the prior for $Ex_{0}$ were reformulated as noisy measurements:
\begin{align*}
Bu_{k}=&Ex_{k+1}-Ax_{k}-Fw_{k}\\
\bar{r}_{0}=&Ex_{0}+e
\end{align*}
where $e$ is a Gaussian zero mean random vector with variance $P_{0}$ and independent of $w_{k},v_{k}$, while $x_{k}$ was viewed as parameters.  The objective was to construct recursively the filtered or predicted estimate given by the conditional mean.  It may be argued, however, that this paradigm shift in viewing $Bu_{k}$ as a measurement is inconsistent with the reality that $Bu_{k}$ is a user defined input that is not random.  Also, the study in \cite{Nikoukhah1999} presented an algorithm for transforming non-causal stochastic descriptor systems into causal systems but did not analyse how to avoid stochastic non-causality which is more meaningful for state estimation problems in practice.  In \cite{Gerdin2006} and \cite{schon2006estimation}, matrix and state variable transformations were used to recast state estimation problems for square causal/non-causal descriptor systems into conventional state space estimation problems.  However, the method results in estimating transformed state variables instead of the original model variables and hence adds a requirement for an inverse transformation at every iteration for finding the estimates.  Moreover, the method was not generalized to non-square descriptor systems.

In this chapter, it is shown that model transformations are not necessary if the system is causal and the algebraic equations are modelled properly to avoid stochastic non-causality.  The analysis is based on examining the various subsystems that descriptor systems can represent using Kronecker canonical transformation.  This canonical form is suitable for extracting the prior on $x_{k}$ for the most general case of the system dynamics \eqref{dts1_ch3} (i.e. causal or non-causal, square or non-square), and is also capable of revealing the necessary assumptions and restrictions needed on the stochastic model and noisy measurements that define a well posed estimation problem.  

\subsection{The Real Kronecker Canonical Form of a Matrix Pencil $\lambda E-A$}
The Kronecker canonical form transformation (KCF) for singular matrix pencils $\lambda E-A$ was developed by the German mathematician Leopold Kronecker in 1890.  This is also often called the generalized Schur decomposition of an arbitrary matrix pencil $\lambda E-A$ and is a generalization of the Jordan canonical form for a square matrix.   

\begin{definition}\cite{gantmakher1959}
The matrix pencil $\lambda E-A$ is said to be singular if $n_{eq}\ne n$ or $det(\lambda E-A)=0~\forall \lambda \in \mathbb{C}$.  Otherwise, if $n_{eq}=n$ and there exist a $\lambda \in \mathbb{C}$ such that $det(\lambda E-A)\ne 0$ then the matrix pencil is called regular.
\end{definition}

\begin{definition}\cite{gantmakher1959}
The matrix pencil $\lambda \tilde{E}-\tilde{A}$ is said to be strictly equivalent to the matrix pencil $\lambda E-A$ for all $\lambda\in \mathbb{C}$ if there exist constant non-singular matrices $P \in \mathbb{C}^{n_{eq}\times n_{eq}}$ and $Q \in \mathbb{C}^{n\times n}$ independent of $\lambda$ such that:
\begin{align}\label{equiv}
\tilde{E}=PEQ,~~\tilde{A}=PAQ
\end{align}
\end{definition}

\begin{definition}
A nonzero vector $x\in \mathbb{C}^{n}$ is a generalized eigenvector of the pair $(E,A)$ if there exists a scalar $\lambda\in \mathbb{C}$, called a generalized eigenvalue such that:
\begin{align*}
(\lambda E-A)x=0
\end{align*}
\end{definition}

\begin{theorem}\label{kcf}\cite{gantmakher1959}, \cite{kunkel2006}
Let $E,A \in \mathbb{R}^{n_{eq}\times n}$.  Then there exists non-singular matrices $P\in \mathbb{R}^{n_{eq}\times n_{eq}}$ and $Q\in \mathbb{R}^{n \times n}$ for all $\lambda \in \mathbb{R}$ such that:
\begin{align}\label{transform}
P(\lambda E-A)Q=\lambda \tilde{E}-\tilde{A}=diag\left(\mathtt{U}_{\epsilon_{0}},\cdots,\mathtt{U}_{\epsilon_{p}},\mathtt{J}_{\rho_{1}},\cdots,\mathtt{J}_{\rho_{r}},\mathtt{N}_{\sigma_{1}},\cdots,\mathtt{N}_{\sigma_{o}},\mathtt{O}_{\eta_{0}},\cdots,\mathtt{O}_{\eta_{q}}\right)
\end{align}
where the matrix blocks are defined as follows:
\begin{enumerate}
\item The block $\mathtt{U}_{\epsilon_{0}}$ correspond to the existence of scalar dependencies between the columns of $\lambda E-A$ and is a zero matrix of size $n_{eq}\times \epsilon_{0}$.  Blocks of the type $\mathtt{U}_{\epsilon_{i}}$ for $i=1,\cdots p$ are the bidiagonal pencil blocks of size $\epsilon_{i}\times (\epsilon_{i}+1)$ and have the form:
\begin{align}\label{under}
\mathtt{U}_{\epsilon_{i}}=\lambda E_{\mathtt{U}_{\epsilon_{i}}}-A_{\mathtt{U}_{\epsilon_{i}}}=\lambda \left[\begin{array}{cccc}0&1&&\\&\ddots & \ddots &\\& &0&1\end{array}\right]-\left[\begin{array}{cccc}1&0&&\\& \ddots & \ddots &\\& &1&0\end{array}\right]
\end{align}
This subsystem has a right null space polynomial vector of the form $[\lambda^{\epsilon_{i}},\lambda^{\epsilon_{i}-1},\cdots,\lambda,1]^{T}$ for any $\lambda$. 

\item $\mathtt{J}_{\rho_{i}}$ are the real Jordan blocks of size $\rho_{i}\times \rho_{i}$ for $i=1,\cdots r$ that correspond to the  generalized eigenvalues of $\lambda E-A$ of the form $(\alpha_{\rho_{i}}-\lambda)^{\rho_{i}}$, with:
\begin{align}\label{jord}
\mathtt{J}_{\rho_{i}}(\alpha_{i})=\lambda E_{\mathtt{J}_{\rho_{i}}}-A_{\mathtt{J}_{\rho_{i}}}=\lambda \left[\begin{array}{cccc}1& &&\\&\ddots &&\\& &\ddots &\\&&&1 \end{array}\right]-\left[\begin{array}{cccc}\alpha_{\rho_{i}}&1&&\\& \ddots & \ddots &\\& &\ddots &1\\&&&\alpha_{\rho_{i}}\end{array}\right]
\end{align}
for real generalized eigenvalues $\alpha_{\rho_{i}}\in \mathbb{R}$ and:

\begin{align}\label{jord}
&\mathtt{J}_{\rho_{i}}(\alpha_{i})=\lambda E_{\mathtt{J}_{\rho_{i}}}-A_{\mathtt{J}_{\rho_{i}}}=\lambda \left[\begin{array}{cccc}1& &&\\&\ddots &&\\& &\ddots &\\&&&1 \end{array}\right]-\left[\begin{array}{cccc}\Delta_{\rho_{i}}&I_{2}&&\\& \ddots & \ddots &\\& &\ddots &I_{2}\\&&&\Delta_{\rho_{i}}\end{array}\right],\nonumber \\
&\Delta_{\rho_{i}}:=\left[\begin{array}{cc} \mu_{\rho_{i}}& \omega_{\rho_{i}} \\ -\omega_{\rho_{i}} & \mu_{\rho_{i}} \end{array} \right]
\end{align}
for complex conjugate generalized eigenvalues $\alpha_{i}=\mu_{i}+j\omega_{i}$, $\bar{\alpha}_{k}=\mu_{i}-j\omega_{i}~\in \mathbb{C}$ with $\omega_{i}>0$.  

\item $\mathtt{N}_{\sigma_{i}}$ are the nilpotent blocks of size $\sigma_{i}\times \sigma_{i}$ for $i=1,\cdots o$ that correspond to the infinite generalized eigenvalues of $\lambda E-A$ with multiplicity $\sigma_{i}$ and have the form:
\begin{align}\label{nilp}
\mathtt{N}_{\sigma_{i}}=\lambda E_{\mathtt{N}_{\sigma_{i}}}-A_{\mathtt{N}_{\sigma_{i}}}=\lambda \left[\begin{array}{cccc}0&1&&\\&\ddots &\ddots &\\& &\ddots &1\\&&&0 \end{array}\right]-\left[\begin{array}{cccc}1& &&\\&\ddots &&\\& &\ddots &\\&&&1 \end{array}\right]
\end{align}
\item The block $\mathtt{O}_{\eta_{0}}$ correspond to the existence of scalar dependencies between the rows of $\lambda E-A$ and is a zero matrix of size $\eta_{0}\times n$.  Blocks $\mathtt{O}_{\eta_{i}}$ are the bidiagonal blocks of size $(\eta_{i}+1)\times \eta_{i}$ for $i=1,\cdots q$ and have the form:
\begin{align}\label{over}
\mathtt{O}_{\eta_{i}}=\lambda E_{\mathtt{O}_{\eta_{i}}}-A_{\mathtt{O}_{\eta_{i}}}=\lambda \left[\begin{array}{ccc}1&&\\0&\ddots &\\&\ddots &1\\&&0\end{array}\right]-\left[\begin{array}{cccc}0&&\\1& \ddots &\\&\ddots &0\\&&1\end{array}\right]
\end{align}
This subsystem has a left null space polynomial vector of $[1,\lambda,\cdots,\lambda^{\eta_{i}}]$ for any $\lambda$.  
\end{enumerate}
\end{theorem}
\begin{proof}
See \cite{gantmakher1959} for full proof of the existence of $P$ and $Q$.  
\end{proof}
Note that the indices $\epsilon_{i},\rho_{i},\sigma_{i}$ and $\eta_{i}$ and the finite generalized eigenvalues $\alpha_{i}$ fully characterize the matrix pencil $\lambda E - A$.  The presence of all these blocks in a pencil reflects the most general case.  The matrices $E,A$ in \eqref{transform} may correspond to a descriptor system model as the one given in \eqref{dts1_ch3}.  In this case, the descriptor system model may contain multiple subsystems, some connected and some disjoint from each other.  Moreover, these subsystems may be under-determined while others over-determined depending on the existence of the blocks $\mathtt{U}_{\epsilon_{i}}$ and $\mathtt{O}_{\eta_{i}}$ in the transformed matrix pencil $\lambda \tilde{E}-\tilde{A}$.   A geometric implementation of the Kronecker canonical decomposition that results in finding real transformation matrices $P\in \mathbb{R}^{n_{eq}\times n_{eq}}$ and $Q \in \mathbb{R}^{n\times n}$ can be found in \cite{berger2012}.  This is important to avoid transforming real random variables $x_{0},w_{k},v_{k}$ to complex random variables which will complicate the analysis otherwise.

\subsection{Transformation of Stochastic Descriptor Systems to KCF}\label{transf}
In order to find the MAP estimate for the state sequence of a stochastic descriptor system \eqref{dts1_ch3} given noisy measurements \eqref{meas_ch3} and initial condition prior \eqref{prior} we need to find the prior on $\{x_{k}\}_{k=0}^{T}$ by transforming \eqref{dts1_ch3} into Kronecker canonical form as follows:
\begin{align}\label{dts2a}
PEQ\tilde{x}_{k+1}=&PAQ\tilde{x}_{k}+PBu_{k}+PFw_{k} \\
y_{k}=&HQ\tilde{x}_{k}+v_{k}
\end{align}
where, 
\begin{align}
\tilde{x}_{k}=Q^{-1}x_{k}=\left[\begin{array}{c}\tilde{x}_{k}^{(\mathtt{U})}\\\tilde{x}_{k}^{(\mathtt{J})}\\\tilde{x}_{k}^{(\mathtt{N})}\\\tilde{x}_{k}^{(\mathtt{O})}\end{array}\right] \label{trans2}
\end{align}
Since the pencil $P(\lambda E-A)Q$ is block diagonal we can compute the solution for each block separately as given by \cite{brull2007}.  In the sequel, we will present this solution in a form suitable for MAP estimation.  We partition the resulting transformed system matrices in \eqref{dts2a} as follows:
\begin{align}\label{mattrans}
\tilde{E}=&PEQ=\left[\begin{array}{cccc}E_{\mathtt{U}}&0&0&0\\ 0&E_{\mathtt{J}}&0&0\\0&0& E_{\mathtt{N}}&0\\ 0&0&0&E_{\mathtt{O}}\end{array}\right],~~\tilde{A}=PAQ=\left[\begin{array}{cccc}A_{\mathtt{U}}&0&0&0\\ 0&A_{\mathtt{J}}&0&0\\0&0& A_{\mathtt{N}}&0\\ 0&0&0&A_{\mathtt{O}}\end{array}\right],\nonumber \\
\tilde{B}=&PB=\left[\begin{array}{c}B_{\mathtt{U}}\\ B_{\mathtt{J}}\\ B_{\mathtt{N}}\\ B_{\mathtt{O}}\end{array}\right],~~\tilde{F}=PF=\left[\begin{array}{c}F_{\mathtt{U}}\\ F_{\mathtt{J}}\\ F_{\mathtt{N}}\\ F_{\mathtt{O}}\end{array}\right],~~\tilde{H}=HQ=\left[\begin{array}{cccc}H_{\mathtt{U}}& H_{\mathtt{J}}& H_{\mathtt{N}}& H_{\mathtt{O}}\end{array}\right]
\end{align}
where,
\begin{align}
E_{\mathtt{U}}=&diag(E_{\mathtt{U}_{\epsilon_{1}}},\cdots,E_{\mathtt{U}_{\epsilon_{p}}}),~~ E_{\mathtt{J}}=diag(E_{\mathtt{J}_{\rho_{1}}},\cdots,E_{\mathtt{J}_{\rho_{r}}}),\cdots etc.\nonumber \\~~
A_{\mathtt{U}}=&diag(A_{\mathtt{U}_{\epsilon_{1}}},\cdots,A_{\mathtt{U}_{\epsilon_{p}}}),~~ A_{\mathtt{J}}=diag(A_{\mathtt{J}_{\rho_{1}}},\cdots,A_{\mathtt{J}_{\rho_{r}}}), \cdots etc.\nonumber
\end{align}
where $diag (\cdot)$ denotes diagonal concatenation of matrices.  Similarly, $B_{\mathtt{U}},\cdots,B_{\mathtt{O}}$ and $F_{\mathtt{U}},\cdots,F_{\mathtt{O}}$ are defined conformally with the rows of $\tilde{E}$ and $\tilde{A}$.  Consequently, \eqref{dts1_ch3} can be expressed using the previous transformation and definitions as follows:
\begin{align}\label{dts2}
\tilde{E}\tilde{x}_{k+1}=&\tilde{A}\tilde{x}_{k}+\tilde{B}u_{k}+\tilde{F}w_{k} \\
y_{k}=&\tilde{H}\tilde{x}_{k}+v_{k}
\end{align}
The solution for each subsystem block will be presented next to examine the prior for each subsystem and to eventually find the MAP estimate for the untransformed variable $x_{k}$.  The analysis and solution of general non-square discrete time descriptor systems was recently conducted in \cite{brull2007}.  We conduct a similar analysis for the stochastic descriptor system \eqref{dts1_ch3}.

\subsubsection{The Under-Determined Subsystem Block}
In the transformed pencil $\lambda \tilde{E}-\tilde{A}$ multiple columns of zeros will occur depending on the number of dependent columns in the original matrix pencil $\lambda E-A$ that differ by a scalar factor.  This corresponds to the existence of the block $\mathtt{U}_{\epsilon_{0}}=0$ with size $n_{eq}\times \epsilon_{0}$.  This implies that there will be descriptor state variables that are unspecified by the descriptor system \eqref{dts1_ch3} and the number of unspecified descriptor states will depend on the number of zero columns in $\lambda \tilde{E}-\tilde{A}$.

For the general case, when $\epsilon_{i}>0$, there will be dependent columns in $\lambda \tilde{E}-\tilde{A}$.  Assuming $\epsilon_{1}>0$, and only the block $\mathtt{U}_{\epsilon_{1}}$ appears in the KCF, the following difference equation can be formed from \eqref{under}:
\begin{align}\label{under1}
E_{\mathtt{U}}\tilde{x}_{k+1}^{(\mathtt{U})}=&A_{\mathtt{U}}\tilde{x}_{k}^{(\mathtt{U})}+B_{\mathtt{U}}u_{k}+F_{\mathtt{U}}w_{k}
\end{align}
Referring to \eqref{under}, this can be written in expanded form as:
\begin{align}\label{xunder}
\left[\begin{array}{c}\tilde{x}_{\mathtt{U},k+1}^{2}\\ \vdots \\ \tilde{x}_{\mathtt{U},k+1}^{\epsilon_{1}+1}\end{array}\right]
=\left[\begin{array}{c} \tilde{x}_{\mathtt{U},k}^{1}\\ \vdots \\ \tilde{x}_{\mathtt{U},k}^{\epsilon_{1}}\end{array}\right]+\left[\begin{array}{c}b_{\mathtt{U}_{1}}\\ \vdots \\ b_{\mathtt{U}_{\epsilon_{1}}}\end{array}\right]u_{k} + \left[\begin{array}{c}f_{\mathtt{U}_{1}}\\ \vdots \\ f_{\mathtt{U}_{\epsilon_{1}}}\end{array}\right]w_{k}
\end{align}
where, $b_{\mathtt{U}_{1}},\cdots$ and $f_{\mathtt{U}_{1}},\cdots$ are formed from the rows of $B_{\mathtt{U}}$ and $F_{\mathtt{U}}$ respectively.  Since $\tilde{x}_{\mathtt{U},k}^{1}$ can not be specified from the stochastic dynamic equations, the subsystem is called an under-determined subsystem.   As a result, we can not find a prior distribution for $\tilde{x}_{\mathtt{U},k}^{1}$ from the transformed stochastic equations.  To reflect our lack of knowledge of this variable, we will assume the following uninformative prior distribution on this random variable $\tilde{x}_{\mathtt{U},k}^{1}$:
\begin{align}
\tilde{x}_{\mathtt{U},k}^{1}\sim \mathtt{N}(\mu_{\mathtt{U}}^{1},q^{2})
\end{align}
where $q>0$ is chosen to be large to make the prior uninformative.  We also assume that $\tilde{x}_{\mathtt{U},k}^{1}$ is independent from the random sequences $w_{k},v_{k}$.  However, we can not estimate $\tilde{x}_{\mathtt{U},k}^{1}$ with this type of prior.  The only way we can estimate this descriptor variable is by having an observation $y_{k}$ that depends on $\tilde{x}_{\mathtt{U},k}^{1}$.  This condition is fulfilled when $[E_{\mathtt{U}}^{T}~~H_{\mathtt{U}}^{T}]^{T}$ is full column rank.  This is known as the estimableness condition given in \cite{Nikoukhah1999}.

\subsubsection{The Over-Determined Subsystem Block}
In the transformed pencil $\lambda \tilde{E}-\tilde{A}$ multiple rows of zeros will occur depending on the number of dependent rows that differ by a scalar factor in the original matrix pencil $\lambda E-A$.  This will corresponds to the existence of the block $\mathtt{O}_{\eta_{0}}=0$ with size $\eta_{0}\times n$.  If $\lambda \tilde{E}-\tilde{A}$ happens to have a row of zeros, then this will correspond to the following difference equation in \eqref{dts2}:
\begin{align}
0=0+B_{\mathtt{O}_{\eta_{i}}}u_{k}+F_{\mathtt{O}_{\eta_{i}}}w_{k}
\end{align}
which imposes constraints on the input and hence is not a well defined stochastic equation since the assumption that $u_{k}$ is deterministic is now invalid.  As a conclusion, for a well defined stochastic model \eqref{dts1_ch3}, we can not have any dependency between the rows of the matrix pencil $\lambda E-A$; i.e. the matrix pencil must be full row rank.  Consequently, in order for the stochastic model \eqref{dts1_ch3} to be well defined, it can not have Kronecker blocks of the form $\mathtt{O}_{\eta_{i}}$.  This is equivalent of having $[E~~A]$ full row rank, which is one of the conditions for a well-posed estimation problem mentioned in \cite{Nikoukhah1999}.

\subsubsection{The Regular Subsystem Block}
The regular subsystem is composed of the Jordan blocks $\mathtt{J}_{\rho}$ and the nilpotent blocks $\mathtt{N}_{\sigma}$ which correspond to the finite and infinite elementary divisors of $\lambda E-A$ respectively.  These two blocks combine to form a square regular descriptor system.  

Assuming $\rho_{1}>0$, then the corresponding difference equation will be:
\begin{align}\label{xjord}
\tilde{x}_{k+1}^{(\mathtt{J})}=&A_{\mathtt{J}}\tilde{x}_{k}^{(\mathtt{J})}+B_{\mathtt{J}}u_{k}+F_{\mathtt{J}}w_{k}
\end{align}
As a result we obtain an ordinary state space difference equation and all state variables can be determined from this subsystem. 

Similarly, if we assume $\sigma_{1}>0$, the nilpotent block appearing in \eqref{nilp} corresponds to the following difference equation:
\begin{align}\label{nilp1}
E_{\mathtt{N}}\tilde{x}_{k+1}^{(\mathtt{N})}=&A_{\mathtt{N}}\tilde{x}_{k}^{(\mathtt{N})}+B_{\mathtt{N}}u_{k}+F_{\mathtt{N}}w_{k}
\end{align}
which can be expanded using \eqref{nilp} as follows:
\begin{align}\label{nilpexp}
\left[\begin{array}{c}\tilde{x}_{\mathtt{N},k+1}^{2}\\\tilde{x}_{\mathtt{N},k+1}^{3}\\ \vdots \\ \tilde{x}_{\mathtt{N},k+1}^{\sigma_{1}}\\0\end{array}\right]
=\left[\begin{array}{c}\tilde{x}_{\mathtt{N},k}^{1}\\ \tilde{x}_{\mathtt{N},k}^{2}\\ \vdots \\ \tilde{x}_{\mathtt{N},k}^{\sigma_{1} -1}\\\tilde{x}_{\mathtt{N},k}^{\sigma_{1}}\end{array}\right]+\left[\begin{array}{c}b_{\mathtt{N}_{1}}\\b_{\mathtt{N}_{c}}\\ \vdots \\ b_{\mathtt{N}_{\sigma_{1}-1}}\\b_{\mathtt{N}_{\sigma}}\end{array}\right]u_{k} + \left[\begin{array}{c}f_{\mathtt{N}_{1}}\\f_{\mathtt{N}_{c}}\\ \vdots \\ f_{\mathtt{N}_{\sigma_{1}-1}}\\f_{\mathtt{N}_{\sigma}}\end{array}\right]w_{k}
\end{align}
We recognize that the matrix $E_{\mathtt{N}}$ is nilpotent of degree $\sigma_{1}$; i.e. $E_{\mathtt{N}}^{i}\ne 0$ for $i<\sigma_{1}$ and $E_{\mathtt{N}}^{i}= 0$ for $i \ge \sigma_{1}$.  As a result, the solution to \eqref{nilpexp} can be expressed as follows:
\begin{align}\label{xnilp}
\tilde{x}_{k}^{(\mathtt{N})}=&E_{\mathtt{N}}\tilde{x}_{k+1}^{(\mathtt{N})}-B_{\mathtt{N}}u_{k}-F_{\mathtt{N}}w_{k}\nonumber \\
=&E_{\mathtt{N}}^{2}\tilde{x}_{k+2}^{(\mathtt{N})}-E_{\mathtt{N}}B_{\mathtt{N}}u_{k+1}-E_{\mathtt{N}}F_{\mathtt{N}}w_{k+1}-B_{\mathtt{N}}u_{k}-F_{\mathtt{N}}w_{k}\nonumber \\  \Rightarrow\tilde{x}_{k}^{(\mathtt{N})}=&-\sum_{i=0}^{\sigma_{1} -1}E_{\mathtt{N}}^{i}B_{\mathtt{N}}u_{k+i}-\sum_{i=0}^{\sigma_{1} -1}E_{\mathtt{N}}^{i}F_{\mathtt{N}}w_{k+i}
\end{align}
This subsystem forms the non-causal equations that correspond to the infinite elementary divisors of$(\lambda E-A)$.  We notice that $\tilde{x}_{k}^{(\mathtt{N})}$ can depend on future values of the input and noise sequences if the system is non-causal.  In order to determine this state, we need to know the future values of the input $u_{k}$ and disturbance sequence $w_{k}$.  The nilpotency of the matrix $E_{\mathtt{N}}$ determines the index of the descriptor model \eqref{dts1_ch3}; i.e. $\nu_{d}=\sigma_{1}$ for time invariant square descriptor models.  We recognize that a high index model; i.e. $\nu_{d}>1$ is not a sufficient condition for having a non-causal model, rather the matrices $B$ and $F$ must also have certain values such that $E_{\mathtt{N}}^{i}B_{\mathtt{N}}\ne 0$ and $E_{\mathtt{N}}^{i}F_{\mathtt{N}}\ne 0$ for $i=1,2,\cdots, \nu_{d}-1$.

Non-causal systems, however, do not exist in reality, unless the variation is with respect to space rather than time.  Techniques for verifying causality and designing the matrix $F$ so that \eqref{dts1_ch3} is causal.


\subsection{The MAP Estimate for Index 1 Causal Descriptor Systems}\label{map}
We have seen that the KCF is capable of performing the following tasks simultaneously:
\begin{enumerate} 
\item Introducing zero column vectors in the transformed matrix pencil that correspond to dependent columns $\lambda E-A$ that differ by a scalar or polynomial factor.  This allows us to determine the descriptor state variables that have no informative prior in the stochastic model upfront.  If $[E_{\mathtt{U}}^{T}~~H_{\mathtt{U}}]^{T}$ is full column rank, then any unspecified states can be estimated from measurements only.
\item Introducing zero row vectors in the transformed matrix pencil that correspond to redundant rows of $\lambda E-A$ that differ by a scalar or polynomial factor.  This allows us to determine if the stochastic model \eqref{dts1_ch3} is well defined as redundant rows will constrain the input sequence and render the estimation problem now well defined.  
\item Determining the Jordan blocks that correspond to the hidden stochastic state-space subsystems in the descriptor model \eqref{dts1_ch3}.
\item Determining the nilpotent blocks that correspond to the hidden non-causal subsystems.
\end{enumerate}
To find the MAP estimate for $x_{k}$, the MAP estimate for $\tilde{x}_{k}$ will be determined first and then inverse transformation will be used (using real transformation matrices $P,Q$ as given in \eqref{trans2}) to find the corresponding value of $\hat{\mathbf{x}}^{map}$.  Based on the above discussion, the following assumptions are needed:

\begin{assumption}{Index 1 Causal Descriptor Systems}\label{assumptions} \\
\begin{enumerate}
\item The matrix $[E~A]$ is full row rank; i.e. there is no dependency between the rows of the matrix pencil $\lambda E-A$ and hence the stochastic model \eqref{dts1_ch3} has a solution to any consistent initial condition.  For example, if the matrices are square, this condition will guarantee that $det(\lambda E-A)\ne 0$ $\forall \lambda \in \mathbb{C}$ which is the condition for system solvability \cite{kunkel2006}.  If the system is rectangular, then the number of rows must be smaller than the number of columns which guarantees existence of a solution to the initial value problem \cite{kunkel2006}. 
\item The matrix $[E^{T}~~H^{T}]^{T}$ is full column rank.  This will enable estimating descriptor states with no informative prior from the stochastic model \eqref{dts1_ch3}.  More precisely, it is required that $[E_{\mathtt{U}_{\epsilon_{i}}}^{T}~~H_{\mathtt{U}_{\epsilon_{i}}}^{T}]^{T}$ be full column rank because only the under-determined subsystems contains the unspecified states as explained earlier.  The two rank conditions are identical since transformation matrices do not alter the rank of the matrices. 
\item The random i.i.d. sequences $w_{k}\sim \mathtt{N}(0,I)$,  $v_{k}\sim \mathtt{N}(0,I)$ and the random variable $\bar{r}_{0}\sim \mathtt{N}(\bar{r}_{0},\mathbf{P}_{0})$ are uncorrelated.
\item The matrix $F$ is full column rank.  This is not a limiting assumption as if it is permitted to redefine the random variables $w_{k}$, then using QR decomposition: \cite{Nikoukhah1999}
\begin{align*}
F=[{F}^{\prime}~~0]\left[\begin{array}{c}Q_{1}\\Q_{2}\end{array}\right],~~~w_{k}^{\prime}=Q_{1}w_{k}
\end{align*}
where, $F^{\prime}$ is full column rank and $w_{k}^{\prime}$ are i.i.d zero mean unit covariance Gaussian vectors because $Q_{1}$ is orthonormal \cite{Nikoukhah1999}.
\item The index of the stochastic descriptor system \eqref{dts1_ch3} is 1 which can be verified using index calculation methods for square descriptor systems as given in \cite{kunkel2006}.  This will also ensure that the system is causal.
\end{enumerate}
\end{assumption}
Consequently, the set of equations that describe the original stochastic descriptor system \eqref{dts1_ch3} and noisy measurements after using KCF transformation and using the above assumptions are as follows:
\begin{subequations}\label{xall}
\begin{align}
E_{\mathtt{U}}\tilde{x}_{k+1}^{\mathtt{U}}=&A_{\mathtt{U}}\tilde{x}_{k}^{\mathtt{U}} +B_{\mathtt{U}}u_{k}+ F_{\mathtt{U}}w_{k} \\
\tilde{x}_{\mathtt{U},k}^{1}=&\mu_{\mathtt{U},k}^{1}+qs_{k} \\
\tilde{x}_{k+1}^{(\mathtt{J})}=&A_{\mathtt{J}}\tilde{x}_{k}^{(\mathtt{J})} +B_{\mathtt{J}}u_{k}+F_{\mathtt{J}}w_{k} \\
\tilde{x}_{k}^{(\mathtt{N})}=&-B_{\mathtt{N}}u_{k}-F_{\mathtt{N}}w_{k} \\
y_{k}=&\tilde{H}\tilde{x}_{k}+v_{k}\label{nilpri}
\end{align}
\end{subequations}
where an uninformative prior was specified for the undetermined state $\tilde{x}_{\mathtt{U},k}^{1}\sim \mathtt{N}(\mu_{\mathtt{U},k}^{1},q^{2})$, with $s_{k}$ is a normally distributed random sequence with zero mean and unit covariance independent from the noise sequences $w_{k},v_{k}$.  All of these equations are explicit in the descriptor state vector.  Since $[E~A]$ is full row rank, we do not have any over-determined subsystem blocks $\mathtt{O}_{\eta{j}}$. 
The MAP estimate of $\{\tilde{x}_{k}\}_{k=0}^{T}$ is obtained from the conditional distribution:
\begin{align}
p_{\mathbf{\tilde{x}}|\mathbf{y}}(\mathbf{\tilde{x}}|\mathbf{y}) \propto & p_{\mathbf{y}|\mathbf{\tilde{x}}}(\mathbf{y}|\mathbf{\tilde{x}})p_{\mathbf{\tilde{x}}}(\mathbf{\tilde{x}})
\end{align}
where $\mathbf{\tilde{x}}=\{\tilde{x}_{k}\}_{k=0}^{T}$.  In the following the subscript for the prior distributions will be omitted for simplicity of notation and will be implied that the distributions are with respect to the random variable $\tilde{x}_{k}$.  From \eqref{xall} and noting the independences between random variables, we recognize that:
\begin{align}
p(\mathbf{\tilde{x}})=&p(\tilde{x}_{0}^{(\mathtt{U})},\tilde{x}_{0}^{(\mathtt{J})},\tilde{x}_{0}^{(\mathtt{N})},\cdots,\tilde{x}_{T}^{(\mathtt{U})},\tilde{x}_{T}^{(\mathtt{J})},\tilde{x}_{T}^{(\mathtt{N})})\nonumber \\
=&p(\tilde{x}_{0}^{(\mathtt{U})},\tilde{x}_{0}^{(\mathtt{J})})p(\tilde{x}_{1}^{(\mathtt{U})},\tilde{x}_{1}^{(\mathtt{J})},\tilde{x}_{0}^{(\mathtt{N})}|\tilde{x}_{0}^{(\mathtt{U})},\tilde{x}_{0}^{(\mathtt{J})})\times \cdots \times p(\tilde{x}_{T}^{(\mathtt{U})},\tilde{x}_{T}^{(\mathtt{J})},\tilde{x}_{T-1}^{(\mathtt{N})}|\tilde{x}_{T-1}^{(\mathtt{U})},\tilde{x}_{T-1}^{(\mathtt{J})})\nonumber \\
=&p(\tilde{x}_{0}^{(\mathtt{U})},\tilde{x}_{0}^{(\mathtt{J})})\prod_{k=0}^{T-1}p(\tilde{x}_{k+1}^{(\mathtt{U})},\tilde{x}_{k+1}^{(\mathtt{J})},\tilde{x}_{k}^{(\mathtt{N})}|\tilde{x}_{k}^{(\mathtt{U})},\tilde{x}_{k}^{(\mathtt{J})}) \nonumber \\
=&p(E_{\mathtt{U}}\tilde{x}_{0}^{(\mathtt{U})},E_{\mathtt{J}}\tilde{x}_{0}^{(\mathtt{J})})p(\tilde{x}_{\mathtt{U},0}^{1})\prod_{k=0}^{T-1}p(E_{\mathtt{U}}\tilde{x}_{k+1}^{(\mathtt{U})},E_{\mathtt{J}}\tilde{x}_{k+1}^{(\mathtt{J})},A_{\mathtt{N}}\tilde{x}_{k}^{(\mathtt{N})}|A_{\mathtt{U}}\tilde{x}_{k}^{(\mathtt{U})},A_{\mathtt{J}}\tilde{x}_{k}^{(\mathtt{J})})p(\tilde{x}_{\mathtt{U},k+1}^{1})\label{priorx}
\end{align}
where the matrix multiplications in the last relationship were introduced by examining the relationships \ref{xunder}, \ref{xjord} and \ref{nilpexp} presented earlier.  More precisely,  $E_{\mathtt{U}}$ is a matrix with a zero vector in the first column and identity matrix in the remaining columns.  Hence, $E_{\mathtt{U}}\tilde{x}_{0}^{(\mathtt{U})}$ is independent from $\tilde{x}_{\mathtt{U},0}^{1}$ and $E_{\mathtt{U}}\tilde{x}_{k+1}^{(\mathtt{U})}$ is independent from $\tilde{x}_{\mathtt{U},k+1}^{1}$.  Also, we note that both $E_{\mathtt{J}}$ and $A_{\mathtt{N}}$ are identity matrices and the multiplication with the random variables have no effect.  Finally, the value of $\tilde{x}_{\mathtt{U},k+1}^{\epsilon + 1}$ is independent from the value of $\tilde{x}_{\mathtt{U},k}^{\epsilon}$ and therefore we may use $A_{\mathtt{U}}\tilde{x}_{k}^{(\mathtt{U})}$ instead of $\tilde{x}_{k}^{(\mathtt{U})}$ in the conditional distribution.  As a result, the conditional distribution in \eqref{priorx} can be obtained using the relations given in \eqref{xall} after variable substitution as follows:
\begin{align*}
p(E_{\mathtt{U}}\tilde{x}_{k+1}^{(\mathtt{U})},E_{\mathtt{J}}\tilde{x}_{k+1}^{(\mathtt{J})},A_{\mathtt{N}}\tilde{x}_{k}^{(\mathtt{N})}|A_{\mathtt{U}}\tilde{x}_{k}^{(\mathtt{U})},A_{\mathtt{J}}\tilde{x}_{k}^{(\mathtt{J})})=p_{\tilde{F}w_{k}}(\zeta)
\end{align*}
where,
\begin{align*}
\zeta = \left(\begin{array}{c}E_{\mathtt{U}}\tilde{x}_{k+1}^{\mathtt{U}}-A_{\mathtt{U}}\tilde{x}_{k}^{\mathtt{U}} +B_{\mathtt{U}}u_{k}\\
\tilde{x}_{k+1}^{(\mathtt{J})}-A_{\mathtt{J}}\tilde{x}_{k}^{(\mathtt{J})} +B_{\mathtt{J}}u_{k}\\
-\tilde{x}_{k}^{(\mathtt{N})}-B_{\mathtt{N}}u_{k}\end{array}\right),~~~
\tilde{F}=PF=\left[\begin{array}{c}F_{\mathtt{U}}\\ F_{\mathtt{J}}\\ F_{\mathtt{N}}\end{array}\right],
\end{align*}
Similarly we may find the distribution for the measurement conditioned on the state as:
\begin{align}
p_{\mathbf{y}|\mathbf{\tilde{x}}}(\mathbf{y}|\mathbf{\tilde{x}})=&
\prod_{k=0}^{T}p_{v_{k}}(y_{k}-\tilde{H}\tilde{x}_{k})
\end{align}
Given the prior for $Ex_{0}$ in \eqref{prior} we need to obtain the prior after after multiplication with $P$ as follows:
\begin{align*}
PEx_{0} \sim & \mathtt{N}(P\bar{r}_{0},PP_{0}P^{T})
\end{align*}
which can be rewritten as:
\begin{align}
\left[\begin{array}{ccc}E_{\mathtt{U}} & 0 & 0 \\ 0 & E_{\mathtt{J}} & 0\\ 0 & 0 & E_{\mathtt{N}}\end{array}\right]\left[\begin{array}{c}\tilde{x}_{0}^{(\mathtt{U})} \\ \tilde{x}_{0}^{(\mathtt{J})} \\ \tilde{x}_{0}^{(\mathtt{N})} \end{array}\right] \sim & \mathtt{N}\left(\left(\begin{array}{c}\bar{r}_{0}^{(\mathtt{U})}\\\bar{r}_{0}^{(\mathtt{J})}\\\bar{r}_{0}^{(\mathtt{N})}\end{array}\right),\mathbf{P}_{0}\right)
\end{align}
where $\mathbf{P}_{0}=PP_{0}P^{T}$.  Note that $E_{\mathtt{N}}=0$ and equation \eqref{nilpri} already determines the prior for $\tilde{x}_{0}^{(\mathtt{N})}$ with mean $-B_{\mathtt{N}}u_{0}$ and variance  $F_{\mathtt{N}}F_{\mathtt{N}}^{T}$.  Consequently, the negative logarithm of the conditional distribution can be written as:
\begin{align}\label{objr1}
&-\log p_{\mathbf{y}|\mathbf{\tilde{x}}}(\mathbf{\tilde{x}},\mathbf{y})-\log p(\mathbf{\tilde{x}}) \propto  \frac{1}{2}\left\|\begin{array}{c}E_{\mathtt{U}}\tilde{x}_{0}^{(\mathtt{U})}-\bar{r}_{0}^{(\mathtt{U})}\\ E_{\mathtt{J}}\tilde{x}_{0}^{(\mathtt{J})}-\bar{r}_{0}^{(\mathtt{J})}\\
0-\bar{r}_{0}^{(\mathtt{N})}    \end{array} \right\|_{\mathbf{P}_{0}}^{2}+\frac{1}{2}\sum_{k=0}^{T}\|y_{k}-\tilde{H}\tilde{x}_{k}\|_{R}^{2}\nonumber \\
& ~~~~~~~+\frac{1}{2}\sum_{k=0}^{T-1}\left\|\begin{array}{c}E_{\mathtt{U}}\tilde{x}_{k+1}^{(\mathtt{U})}-A_{\mathtt{U}}\tilde{x}_{k}^{(\mathtt{U})}-B_{\mathtt{U}}u_{k}
 \\
E_{\mathtt{J}}\tilde{x}_{k+1}^{(\mathtt{J})}-A_{\mathtt{J}}\tilde{x}_{k}^{(\mathtt{J})}-B_{\mathtt{J}}u_{k} \\
0-\tilde{x}_{k}^{(\mathtt{N})}-B_{\mathtt{N}}u_{k}\end{array}\right\|_{\mathbf{F}}^{2}+\frac{1}{2q^{2}}\sum_{k=0}^{T}\|\tilde{x}_{\mathtt{U},k}^{1}-\mu_{\mathtt{U},k}^{1}\|^{2}
\end{align}
where, $\mathbf{F}=\tilde{F}\tilde{F}^{T}$.  Notice that $\tilde{F}=PF$ is full column rank by the assumption that $F$ is full column rank.  Hence, $\tilde{F}\tilde{F}^{T}$ is non-singular and positive definite.  Taking the limit as $q \rightarrow \infty$ (to reflect our lack of prior for $\tilde{x}_{\mathtt{U},k}^{1}$) and using the relations for $PEQ,~PAQ$ and $PB$ in \eqref{mattrans} will result in  the following objective function:
\begin{align}
\hat{\mathbf{x}}^{map}=&\mbox{arg}\min_{\mathbf{x}}\frac{1}{2}\|PEx_{0}-P\bar{r}_{0}\|_{\mathbf{P}_{0}}^{2}+\frac{1}{2}\sum_{k=0}^{T}\|y_{k}-Hx_{k}\|_{R}^{2}+\frac{1}{2}\sum_{k=0}^{T-1}\|PEx_{k+1}-PAx_{k}-PBu_{k}\|_{\mathbf{F}}^{2}\nonumber\\
=&\mbox{arg}\min_{\mathbf{x}}\frac{1}{2}\|Ex_{0}-\bar{r}_{0}\|_{P_{0}}^{2}+\frac{1}{2}\sum_{k=0}^{T}\|y_{k}-Hx_{k}\|_{R}^{2}
+\frac{1}{2}\sum_{k=0}^{T-1}\|Ex_{k+1}-Ax_{k}-Bu_{k}\|_{Q}^{2} \label{uncorobj}
\end{align}
where, $P^{T}\mathbf{P}_{0}P=P_{0}$ and $Q=P^{T}\mathbf{F}P=FF^{T}$.  Hence, the MAP estimate for $x_{k}$ for causal index 1 descriptor systems of the form \eqref{dts1_ch3},\eqref{meas_ch3} can be found directly from the system matrices with no need of any transformation.  Solving this minimization problem is identical to solving the constrained maximum likelihood objective function derived in \cite{Nikoukhah1992} and \cite{Nikoukhah1999} by viewing the initial condition and input sequence as noisy measurements.  Hence, this establishes that the MAP and ML estimates are identical for state estimation problems that involve causal descriptor systems of index 1.

\newpage
\bibliographystyle{IEEEtran} 
\bibliography{references} 

\end{document}